\documentclass[12pt]{amsart}
\usepackage{amsmath,amsfonts,amssymb} \usepackage{color}
\usepackage[all,cmtip]{xy}
\usepackage{graphicx}
 \usepackage{youngtab}

\newcommand{\ba}{{\bf a}}

  \newcommand{\cF}{\mathcal{F}} 
  \newcommand{\bp}{{\bf p}}

\newcommand{\bT}{{\bf T}}
 
 \newcommand{\cH}{\mathcal{H}}

   \newcommand{\cE}{\mathcal{E}}
  
\newcommand{\fgl}{\mathfrak{gl}}

 \newcommand{\bZ}{\mathbb{Z}}
 \newcommand{\bC}{\mathbb{C}}
 \newcommand{\pd}{\partial}
\newcommand{\Mbar}{\overline{\mathcal M}}
 
\newcommand{\vac}{|0\rangle} \newcommand{\lvac}{\langle 0|}

 \DeclareMathOperator{\Gr}{Gr}  \DeclareMathOperator{\SL}{SL}
 \DeclareMathOperator{\GL}{GL}

\newcommand{\be}{\begin{equation}}
\newcommand{\ee}{\end{equation}}
\newcommand{\bea}{\begin{eqnarray}}
\newcommand{\eea}{\end{eqnarray}}
\newcommand{\ben}{\begin{eqnarray*}}
\newcommand{\een}{\end{eqnarray*}}

\newcommand{\half}{\frac{1}{2}}

\newtheorem{cor}{Corollary}[section]

\newtheorem{lem}[cor]{Lemma}
 
 \newtheorem{thm}[cor]{Theorem}
\theoremstyle{remark}

 \newtheorem{rmk}[cor]{Remark}

\definecolor{A}{rgb}{.75,1,.75}

\definecolor{green}{rgb}{0,1,0}
\definecolor{yellow}{rgb}{1,1,0}
\definecolor{orange}{rgb}{1,.7,0}
\definecolor{red}{rgb}{1,0,0}
\definecolor{white}{rgb}{1,1,1}

 \makeindex
\begin{document}
\title
{Emergent Geometry of KP Hierarchy. II.}

\author{Jian Zhou}
\address{Department of Mathematical Sciences\\Tsinghua University\\Beijng, 100084, China}
\email{jzhou@math.tsinghua.edu.cn}

\begin{abstract}
We elaborate on a construction of quantum LG superpotential associated
to a tau-function of the KP hierarchy in the case that resulting quantum spectral curve
lies in the quantum two-torus.
This construction is applied to Hurwitz numbers,
one-legged topological vertex and resolved conifold with external D-brane
to give a natural explanation of some earlier work on the relevant quantum curves.
\end{abstract}

\maketitle

\section{Introduction}

In the first part of this paper \cite{Zhou},
we have combined the theory \cite{Tak-Tak} of the twistor data of tau-functions
of the KP hierarchy
with the theory of string equations \cite{Sch}
to obtain a theory of quantum deformation theory of the relevant quantum spectral curves
over the big phase space.
The main idea is that given a tau-function of the KP hierarchy,
one can construct a Lax operator $L$ and Orlov-Schulman operator $M$,
such that $[L, M] = \hbar$.
Furthermore,
there is  an action operator $\hat{S}$ that relates $L$ and $M$
via a quantum Hamilton-Jacobi equation.
A twistor data is given by a pair of differential operators $P$ and $Q$
expressible in terms of $L$ and $M$,
such that $[P, Q] = \hbar$.
When the tau-function is the partition function of Witten's r-spin curves,
this produces a kind of emergent mirror symmetry
that relates some topological matters coupled with 2D topological gravity
to some quantum mechanical system with special actions.

In \cite{Zhou} we focused on the case that at least one of the differential operators, say $P$, given by the twistor
data is of finite order.
We referred to the expression of $P$ in terms of $L$ and $M$ as the quantum Landau-Gingzburg superpotential.
However,
we did not address the issue of how to make the choice for $P$.
Clearly,
if $(P, Q)$ is a twistor data,
so is $(aP +b Q, cP+dQ)$ for any
$$\begin{pmatrix}
a & b \\ c & d
\end{pmatrix} \in \SL_2(\bZ).
$$
Nevetheless in examples there do seem to be some preferred choices.
The resulting quantum spectral curves live the quantum plane.
In the end of Part I we mentioned several directions for generalizations,
one of them being the generalization to the case
when both operators $P$ and $Q$ are of infinite orders.
We will do that in this paper.
We will see that quantum spectral curve more naturally live
in the quantum two-torus in this case.
Furthermore,
we will relax the requirement that $[P, Q] = \hbar$.
So the choice of operators $P$ and $Q$ becomes a more serious issue.
We deal with it as follows.
As in \cite{Sch} we will use the Kac-Schwarz operators $\hat{P}$ and $\hat{Q}$ that characterizes
the element $V$ in Sato Grassmannian corresponding to the tau-function:
\begin{align}
\hat{P}_0 V & \subset V, & \hat{Q}_0 V & \subset V.
\end{align}
We will require that $V$ is spanned by a sequence of series $\{\varphi_n\}_{n \geq 0}$,
such that
\bea
&& \hat{P}_0 \varphi_0 = 0, \\
&& \hat{Q}_0 \varphi_n = \varphi_{n+1}, \;\; n \geq 0.
\eea
The operators $P$ and $Q$ related to $\hat{P}_0$ and $\hat{Q}_0$ respectively
by some Laplace transform and deformation by the KP flow.

As applications,
we will treat the three cases of tau-functions of the KP hierarchy
arising in topological string theory:  (1) Hurwitz numbers and linear Hodge integrals related by ELSV formula,
(2) triple Hodge integrals in Mari\~no-Vafa formula (the one-legged topological vertex),
and  (3) the resolved conifold with one external D-brane.
These three cases are related to each other in the sense that by taking suitable limit,
one gets (2) from (3) and (1) from (2).
The quantum curves (on the small phase spaces) of these three cases were treated in \cite{Zhou-Quant}
from the point of view of Eynard-Orantin recursion \cite{Guk-Sul},
we now treat their deformations over the big phase space
from the point of view of emergent geometry of the KP hierarchy.

We arrange the rest of this paper as follows.
In Section \ref{sec:KS} we explain how to obtain operators $P$ and $Q$ from
Kac-Schwarz operators based on Sato grassmannian.
In the next three Sections we discuss the applications
to tau-functions given by Hurwitz numbers,
one-legged topological vertex,
and resolved conifold with one outer D-brane, respectively,
and compare the results with earlier work on quantum curves for these three cases \cite{Zhou-Quant}.
We conclude the paper by remarking that in most of our examples
the quantum spectral curves live in the quantum complex 2-torus.

\section{Kac-Schwarz Operators and W-Constraints}

\label{sec:KS}

\subsection{Big cell of Sato Grassmannian}

Denote by $\cE$ the algebra of pseudo-differential operators of the form
\be
\sum_{n \geq N} a_n(x, \hbar) (\hbar \pd_x)^n,
\ee
where each $a_n(x, \hbar) = \sum_{m\in \bZ} a_m(\hbar) x^m$ is a Laurent series in $x$.
Then $\cH = \cE/\cE x$ is isomorphic as a vector space to the set
$\bC((\pd_x^{-1}))$,
the space of pseudo-differential operators with constant coefficients.
The algebra $\cE$ naturally acts on $\cH$.
This is because from
\be \label{eqn:Higher Heisenberg}
\pd_x^n x = n \pd_x^{n-1} + x\pd_x^n,
\ee
one can get:
\be
x^m \pd_x^n \cdot \pd_x^l
= (-1)^m m! \binom{n+l}{m}  \pd_x^{n+l-m} \pmod{\cE x}.
\ee
After the Laplace transform,
\be
\cH \cong H = \bC((z^{-1})),
\ee
and the action of $x^m \pd_x^n$ transformed to the action of $(-\pd_z)^m z^n$.

Let $H_+ = \bC[z]$ and $H_- = z^{-1} \bC[[z^{-1}]]$,
One has a decomposition $H = H_+ \oplus H_-$.
Let $\pi: H \to H_+$ be the projection.
The big cell $\Gr^{(0)}$ of Sato grassmannian
consists of linear subspaces $V$ of $H$ such that
$\pi_+|_V: V \to H_+$ are isomorphisms.

\subsection{Boson-fermion correspondence}

The fermionic Fock space $\cF$ is spanned by expressions of the form
\be
z^{k_1} \wedge z^{k_2} \wedge \cdots,
\ee
where $k_1 < k_2 < \cdots$ is a sequence of integers,
such that $k_n -n +1 =0$ for $n \gg 0$.
One can define the fermionic charge of this expression,
and use the charge zero expressions to span a subspace $\cF^{(0)}$ of $\cF$.
The vacuum vector of this space is
\be
\vac = z^0 \wedge z^1 \wedge \cdots z^n \wedge \cdots.
\ee

Denote by $\fgl(\infty)$ the Lie algebra of differential operators in
the variable $z$ with Laurent coefficients.
Every element of this algebra acts
as a linear operator acting on the space $H$:
\be
z^l \pd_z^m (z^k) = m!\binom{k}{m} z^{k+l-m}.
\ee
In particular,
when $l=m$,
\be
z^m \pd_z^m (z^k) = m!\binom{k}{m} z^{k}.
\ee
To every operator $A \in \fgl(\infty)$,
following Kazarian \cite{Kaz}
we define an operator $\hat{A}$ acting on the fermionic Fock space
as follows:
\ben
&& \widehat{z^l \pd_z^m}(z^{k_1}
\wedge z^{k_2} \wedge \cdots ) \\
& = & m!(\binom{k_1}{m} - \delta_{l,m} \binom{1}{m}) z^{k_1+l-m} \wedge z^{k_2} \wedge
z^{k_3} \wedge \cdots \\
& + & z^{k_1} \wedge m!(\binom{k_2}{m} - \delta_{l,m} \binom{2}{m}) z^{k_2+l-m}
) \wedge z^{k_3} \wedge \cdots \\
& + & z^{k_1} \wedge z^{k_2} \wedge m!(\binom{k_3}{m} - \delta_{l,m} \binom{3}{m}) z^{k_3+l-m} \wedge
\cdots + \cdots,
\een
this regularized action defines a central extension $\widehat{\fgl(\infty)}$
of $\fgl(\infty)$.
Write
\be
\alpha_m = \widehat{z^m}.
\ee
The boson-fermion correspondence is a map
$\cF^{(0)} \to \Lambda =\bC[[p_1, p_2, \dots]]$ given by:
\be
|\ba\rangle \in \cF^{(0)} \mapsto
\lvac \exp (\sum_{n=1}^\infty \frac{1}{n} p_n \alpha_n) |\ba\rangle.
\ee
Under this correspondence,
\be
\alpha_m = \widehat{z^m} = \begin{cases}
m \frac{\pd}{\pd p_m}, & \text{if $m > 0$}, \\
0, & \text{if $m =0$}, \\
p_{-m} \cdot, & \text{if $m < 0$}.
\end{cases}
\ee
In the literature,
other examples of operators on $H$ corresponding to
important operators on $\Lambda$ are known \cite{Kac-Sch, Kaz}.
Let
\bea
&& L_m = \half \sum_{j+k=m} :\alpha_j \alpha_k:, \\
&& K_m = \frac{1}{3!} \sum_{j+k+l=m} :\alpha_j\alpha_k\alpha_l:.
\eea
Then one has
\bea
&& L_m = -(z^m (z \pd_z + \frac{m-1}{2}))^{\wedge}, \\
&& K_m
=  \biggl(z^m(\frac{1}{2} (z\pd_z)^2 + \frac{m-1}{2} z\pd_z + \frac{(m-1)(m-2)}{12})
\biggr)^{\wedge}.
\eea
In particular,
the operator $K_0$ is the cut-and-join operator \cite{GJV}:
\be
K_0 = \half \sum_{m,n=1}^\infty ((m+n)p_mp_n \frac{\pd}{\pd p_{m+n}}
+ mn p_{m+n} \frac{\pd^2}{\pd p_m\pd p_n}).
\ee
As noted in \cite{Kaz},
\be
K_0 = \half \biggl( (z\pd_z - \half)^2\biggr)^{\wedge}.
\ee

\subsection{Tau-function associated to elements of Sato Grassmannian}

Let
\be
\widehat{\GL(\infty)} = \{e^{\hat{A}_1} \cdots e^{\hat{A}_n} \;| \; A_1, \dots, A_n \in
\fgl(\infty)\}.
\ee
It acts on $\cF^{(0)}$.
By Sato \cite{Sato},
$\widehat{\GL(\infty)}\vac$ is the space of KP $\tau$-functions.
For an element $g\in \widehat{\GL(\infty)}$,
suppose that
\be
g\vac = \varphi_0 \wedge \varphi_1 \wedge \cdots,
\ee
where $\varphi_n = z^n + $ lower order terms,
then the element $V$ corresponding to $g\vac$ is spanned by
$ \{\varphi_n \}_{n =0, , 2, \dots}$.

\subsection{Kac-Schwarz operators}

Given a tau-function,
one define the wave function
by Sato's formula:
\be \label{eqn:Sato}
w(\bT; \xi) = \exp \biggl(\hbar \sum_{n=1}^\infty  T_n \xi^n \biggr) \cdot
\frac{\tau(T_1-\hbar\xi^{-1}, T_2-\hbar\frac{1}{2} \xi^{-2}, \dots; \hbar)}{\tau(T_1, T_2, \dots; \hbar)}.
\ee
Rewrite $w$ in the following form:
\be \label{eqn:w}
w = \exp \biggl( \hbar^{-1} \sum_{n=1}^\infty T_n \xi^n \biggr) \cdot
\biggl( 1 + \frac{w_1}{\xi} + \frac{w_2}{\xi^2} + \cdots \biggr).
\ee
The dressing operator $W$ is defined by:
\be \label{eqn:W}
W:= 1 + \sum_{n=1}^\infty w_j \pd_x^{-j}.
\ee
The Lax operator $L$ and the Orlov-Schulman operators are defined by:
\bea
&& L : = W \circ \hbar \pd_x \circ W^{-1},  \label{eqn:Lax} \\
&& M : = W\biggl( \sum_{n=1}^\infty n T_n (\hbar \pd_x)^{n-1} \biggr) W^{-1}.
\eea
One can use the dressing operator to
associate an element $V$ corresponding to the tau-function as follows:
\be
V = W^{-1} H_+.
\ee
By a pair $(\hat{P}_0, \hat{Q}_0)$ of Kac-Schwarz operators we mean two differential operators in $z$ such that
\begin{align}
\hat{P}_0V & \subset V, & \hat{Q}_0 V & \subset V,
\end{align}
and furthermore,
$V$ has an admissible basis $\{\varphi_n\}_{n \geq 0}$ such that
\bea
&& \hat{P}_0 \varphi_0 = 0, \\
&& \hat{Q}_0 \varphi_n = \varphi_{n=1}, \;\; n \geq 0.
\eea
We denote by $P_0$ and $Q_0$ the differential operators corresponding to $P_0$ and $Q_0$ respectively,
and let $P$ and $Q$ are differential operators obtained by KP flows with $P_0$ and $Q_0$ as
initial values respectively.

One can convert the conditions that
\be \label{eqn:V}
\hat{P}_0^k \hat{Q}_0^lV \subset V, \;\;\; k, l \geq 0,
\ee
into constraints on the tau-function.
These will be called the $W$-constraints.

\section{Tau-Function from Hurwitz numbers and Linear Hodge Integrals}

\subsection{Operator formalism of Hurwitz numbers and liner Hodge integrals}

Let $Z(\bp; \lambda)$ be the generating series of disconnected Hurwitz numbers.
By ELSV formula it is also the generating series of the following special linear Hodge integrals
\cite{Zhou-Hurwitz}:
\be
\frac{1}{z_\mu} \prod_{j=1}^{l(\mu)} \frac{\mu_j^{\mu_j+1}}{\mu_j!}
\cdot \int_{\Mbar_{g, l(\mu)}} \frac{\Lambda^\vee(1)}{\prod_{j=1}^{\mu_j} (1-\mu_j \psi_j)}.
\ee
It satisfies the cut-and-join equation \cite{GJV}:
\be
\frac{\pd}{\pd \lambda} Z(\bp; \lambda) = K_0 Z(\bp; \lambda)
\ee
and the following initial condition
\be
Z(\bp; \lambda) = e^{p_1},
\ee
and so one has \cite{Kaz-Lan}:
\be
Z(\bp; \lambda) = e^{\lambda K_0} e^{p_1}.
\ee

\subsection{Kac-Schwarz operators for Hurwitz numbers}
The corresponding element in $\Gr^{(0)}$ is spanned by the following series
\be
\begin{split}
\varphi_j = & e^{\frac{\lambda}{2}[(z\pd_z+1/2)^2- (j+1/2)^2]} e^{1/z}z^j \\
= & \sum_{n=0}^\infty e^{\frac{\lambda}{2}[(j-n+1/2)^2-(j+1/2)^2]} \frac{z^{j-n}}{n!},
\end{split}
\ee
they can be written as:
\be \label{eqn:Basis}
\varphi_j
= \sum_{n=0}^\infty e^{n(n-(2j+1))\lambda /2]} \frac{z^{j-n}}{n!}.
\ee
In particular,
\be
\varphi_0 = \sum_{n=0}^\infty e^{n(n-1)\lambda/2} \frac{z^{-n}}{n!}.
\ee
This matches with \cite[(13)]{Zhou-Quant}.
Define two operators $\hat{P}_0, \hat{Q}_0$ as follows:
\bea
&& \hat{P}_0 = z \pd_z + z^{-1} e^{-\lambda z\pd_z}, \\
&& \hat{Q}_0 = z e^{\lambda z\pd_z}.
\eea
Then one has \cite{Ale}:
\bea
&& [\hat{P}_0, \hat{Q}_0] = \hat{Q}_0, \\
&& \hat{Q}_0 \varphi_j = e^{j\lambda} \varphi_{j+1},  \\
&& \hat{P}_0 \varphi_j = j \varphi_j.
\eea
for $j =0, 1, 2, \dots$.
These can be checked directly from \eqref{eqn:Basis}.
In particular,
\be
\hat{P}_0 \varphi_0 = 0.
\ee
This matches with \cite[(17)]{Zhou-Quant}.

\section{Tau-Function from One-Partition Triple Hodge Integrals}

\subsection{One-partition triple Hodge Integrals}

For a partition $\mu$,
consider the following generating series of Hodge integrals:
\begin{eqnarray*}
G_{\mu}(r; \lambda)
& = & - \frac{\sqrt{-1}^{l(\mu)}}{z_{\mu}}
\cdot \left[r(r+1)\right]^{l(\mu)-1}
\cdot \prod_{i=1}^{l(\mu)} \frac{\prod_{a=1}^{\mu_i-1}
\left( \mu_ir + a \right)}{\mu_i!}\\
&& \cdot \sum_{g \geq 0} \lambda^{2g-2} \int_{\Mbar_{g, l(\mu)}}
\frac{\Lambda_{g}^{\vee}(1)\Lambda^{\vee}_{g}(r)\Lambda_{g}^{\vee}(-1 - r)}
{\prod_{i=1}^{l(\mu)} \frac{1}{\mu_i} \left(\frac{1}{\mu_i} - \psi_i\right)},
\end{eqnarray*}
and their generating series:
\begin{eqnarray*}
G^{\bullet}(r; \lambda; \bp) & = &\exp \left(
\sum_{\mu} G_{\mu}(r; \lambda)p_{\mu}\right).
\end{eqnarray*}
We have proved in \cite{Zhou-Hodge} that $G^\bullet$ is a tau-function of the KP hierarchy.
This is established by using
Mari\~no-Vafa formula \cite{Mar-Vaf, LLZ, Oko-Pan} to get:
\be
\begin{split}
& G^{\bullet}(r; \lambda; \bp) \\
= & \langle 0| \exp\left( \sum_{n > 0} \frac{p_n}{n} \alpha_{n}\right)
q^{(r+1)K_0}
 \exp\left( \sum_{n > 0} \frac{(-1)^{n-1}}{q^{n/2} -q^{-n/2}} \frac{\alpha_{-n}}{n} \right)|0\rangle,
\end{split}
\ee
where $q= e^{\sqrt{-1} \lambda}$.

The corresponding element in $\Gr^{(0)}$ is spanned by the following series
\be
\begin{split}
\varphi_j = &
\exp \biggl(\frac{\sqrt{-1}(r+1)\lambda}{2}[(z\pd_z+1/2)^2- (j+1/2)^2]\biggr) \\
&
\exp\left( \sum_{n > 0} \frac{(-1)^{n-1}}{n(q^{n/2} -q^{-n/2})} z^{-n}\right) z^j.
\end{split}
\ee

\begin{lem}
The series $\varphi_j$ can be explicitly written as:
\be \label{eqn:Basis}
\varphi_j
=  \sum_{n=0}^\infty
\frac{(-1)^n q^{(r+1)n(n-2j-1)/2-n(n-1)/4}}{[n]!} z^{j-n}.
\ee
In particular,
\be
\varphi_0 = \sum_{n=0}^\infty \frac{q^{(r+1)n(n-1)/2+n/2}}{\prod_{k=1}^n (1- q^k)} z^{-n}.
\ee
This matches with \cite[(26)]{Zhou-Quant}.
\end{lem}

\begin{proof}
Note we have:
\ben
&&  \exp\left( \sum_{n > 0} \frac{(-1)^{n-1}}{n(q^{n/2} -q^{-n/2})} z^{-n} \right) \\
& = & \exp\left( \sum_{n > 0} \frac{(-1)^n}{n}
\sum_{m=1}^\infty q^{(m-1/2)n} z^{-n} \right) \\
& = & \prod_{m=1}^\infty \frac{1}{1+q^{m-1/2}/z}
= \sum_{n=0}^\infty \frac{(-1)^n q^{n/2}}{\prod_{k=1}^n (1-q^k)} z^{-n},
\een
where $[n]!=\prod_{j=1}^n (q^{j/2} - q^{-j/2})$.
Here in the last equality we have used:
\be
\prod_{m=1}^\infty \frac{1}{1-q^m x}
= 1 + \sum_{n\geq 1} \frac{q^n}{\prod_{j=1}^n (1-q^j)}x^n.
\ee
Now we have:
\ben
\varphi_j & = &
\exp \biggl(\frac{\sqrt{-1}(r+1)\lambda}{2}[(z\pd_z+1/2)^2- (j+1/2)^2]\biggr) \\
&&   \sum_{n=0}^\infty \frac{(-1)^n q^{n/2}}{\prod_{k=1}^n (1-q^k)} z^{j-n} \\
& = & \sum_{n=0}^\infty
\exp \biggl(\frac{\sqrt{-1}(r+1)\lambda}{2}[(j-n+1/2)^2- (j+1/2)^2]\biggr) \\
&& \cdot \frac{(-1)^n q^{n/2}}{\prod_{k=1}^n (1-q^k)} z^{j-n} \\
& = & \sum_{n=0}^\infty \frac{(-1)^n q^{(r+1)n(n-2j-1)/2+n/2}}{\prod_{k=1}^n (1-q^k)} z^{j-n}.
\een
\end{proof}

\begin{thm}
Define two operators $\hat{P}, \hat{Q}$ as follows:
\bea
&& \hat{P}_0 = 1 - q^{-z \pd_z} - q^{1/2} z^{-1} q^{-(r+1) z\pd_z}, \\
&& \hat{Q}_0 = z q^{(r+1) z \pd_z}.
\eea
Then one has:
\bea
&& \hat{Q}_0 \varphi_j = q^{(r+1)j} \varphi_{j+1},  \\
&& \hat{P}_0 \varphi_j = (1-q^{-j}) \varphi_j - q^{1/2- (r+1)j}(1-q^{-j}) \varphi_{j-1},
\eea
for $j =0, 1, 2, \dots$,
where $\varphi_{-1} = 0$.
In particular,
\be
\hat{P}_0 \varphi_0 = 0.
\ee
This matches with \cite[(32)]{Zhou-Quant}.
Furthermore,
these operators satisfy the following commutation relations:
\be \label{eqn:[P,Q]}
[\hat{P}_0, \hat{Q}_0 ] = (1 -q^{-1}) z q^{r z \pd_z}.
\ee
\end{thm}

\begin{proof}
These can be checked directly from \eqref{eqn:Basis}.
\ben
\hat{Q}_0 \varphi_j
& = & z q^{(r+1) z \pd_z} \sum_{n=0}^\infty
\frac{q^{(r+1)n(n-2j-1)/2+n/2}}{\prod_{k=1}^n (1-q^k)} z^{j-n} \\
& = & \sum_{n=0}^\infty
\frac{q^{(r+1)n(n-2j-1)/2+n/2}}{\prod_{k=1}^n (1-q^k)}
q^{(r+1)(j-n)} z^{j+1-n} \\
& = & q^{(r+1)j} \sum_{n=0}^\infty
\frac{q^{(r+1)n(n-2(j+1)-1)/2+n/2}}{\prod_{k=1}^n (1-q^k)} z^{j+1-n} \\
& = & q^{(r+1)j} \varphi_{j+1},
\een
and
\ben
\hat{P}_0 \varphi_j
& = & \sum_{n=0}^\infty
\frac{q^{(r+1)n(n-2j-1)/2+n/2}}{\prod_{k=1}^n (1-q^k)} z^{j-n} \\
& - & \sum_{n=0}^\infty
\frac{q^{(r+1)n(n-2j-1)/2+n/2}}{\prod_{k=1}^n (1-q^k)} q^{n-j} z^{j-n} \\
& - & q^{1/2} z^{-1} \sum_{n=0}^\infty
\frac{q^{(r+1)n(n-2j-1)/2+n/2}}{\prod_{k=1}^n (1-q^k)} q^{(r+1)n - (r+1)j} z^{j-n}\\
& = & \sum_{n=0}^\infty
\frac{q^{(r+1)n(n-2j-1)/2+n/2}}{\prod_{k=1}^n (1-q^k)} z^{j-n} \\
& - & q^{-j} \sum_{n=0}^\infty
\frac{q^{(r+1)n(n-2j-1)/2+n/2}}{\prod_{k=1}^n (1-q^k)}  z^{j-n} \\
& + & q^{-j} \sum_{n=1}^\infty
\frac{q^{(r+1)n(n-2j-1)/2+n/2}}{\prod_{k=1}^{n-1} (1-q^k)}  z^{j-n} \\
& - & q^{1/2- (r+1)j}  \sum_{n=0}^\infty
\frac{q^{(r+1)n(n-2(j-1)-1)/2+n/2}}{\prod_{k=1}^n (1-q^k)} z^{j-1-n} \\
& = & (1-q^{-j}) \varphi_j - q^{1/2- (r+1)j}(1-q^{-j}) \varphi_{j-1}.
\een
By induction one can show that:
\be
[(z\pd_z)^n, z] = z ((1+z\pd_z)^n- (z\pd_z)^n),
\ee
it follows that
\be
[e^{t z\pd_z}, z] = (e^t-1) z e^{t z\pd_z}.
\ee
Therefore,
\ben
[\hat{P}_0, \hat{Q}_0 ]
& = & [1 - q^{-z \pd_z} - q^{1/2} z^{-1} q^{-(r+1)  z\pd_z},
z q^{(r+1)  z \pd_z} ] \\
& = & -[q^{-z \pd_z}, z] q^{(r+1)  z \pd_z}
-q^{1/2} [z^{-1} q^{-(r+1)  z\pd_z},
z q^{(r+1)  z \pd_z} ]  \\
& = & -[q^{-z \pd_z}, z] q^{(r+1)  z \pd_z}
-q^{1/2} [z^{-1} q^{-(r+1)  z\pd_z}, z] q^{(r+1)  z \pd_z}   \\
& - & q^{1/2} z [z^{-1} q^{-(r+1)  z\pd_z}, q^{(r+1)  z \pd_z} ]   \\
& = & -[q^{-z \pd_z}, z] q^{(r+1)  z \pd_z}
-q^{1/2} z^{-1} [q^{-(r+1)  z\pd_z}, z] q^{(r+1)  z \pd_z}   \\
& - & q^{1/2} z [z^{-1} , q^{(r+1)  z \pd_z} ] q^{-(r+1)  z\pd_z}  \\
& = & - (q^{-1} - 1) z q^{r z \pd_z} -q^{1/2}  (q^{-(r+1)}-1)
+ q^{1/2} (q^{-(r+1)}-1) \\
& = & (1 -q^{-1}) z q^{r z \pd_z}.
\een
\end{proof}

\begin{rmk}
It is interesting to consider the Lie algebra generated by
the two operators $\hat{P}_0$ and $\hat{Q}_0$.
\end{rmk}

\section{Tau-Function of KP Hierarchy from Resolved Conifold}

\subsection{Open string amplitudes with one outer brane}

Let us recall some results from \cite{Zhou-MV}.
One can compute the open string amplitude for the resolved conifold with one outer brane and framing $a$
by the theory of the topological vertex \cite{AKMV, LLLZ}.
There are two possibilities,
corresponding to two different toric diagrams as follows:
$$
\xy
(0,10); (0,0), **@{-}; (-7, -7), **@{-}; (-10, -4), **@{.};
(0,0); (20,0), **@{-}; (27,7), **@{-}; (20,0); (20, -10), **@{-};
(-5, -2)*+{\mu}; (5,1.5)*+{\nu^t}; (15,1)*+{\nu}; (11,-14)*+{(i)};
\endxy
\qquad
\xy
(-4,5); (0,10), **@{.}; (0,0), **@{-}; (-7, -7), **@{-};
(0,0); (20,0), **@{-}; (27,7), **@{-}; (20,0); (20, -10), **@{-};
(-2, 3)*+{\mu}; (5,1.5)*+{\nu^t}; (15,1)*+{\nu}; (11,-14)*+{(ii)};
\endxy
$$
\begin{center}
{\bf Figure 1.}
\end{center}
The normalized open string amplitude of the resolved conifold with one outer brane as in Figure 1(a)
and framing $a \in \bZ$ is given by:
\be \label{eqn:ZaOperator1}
\begin{split}
& \hat{Z}^{(a)}(\lambda;t; \bp) \\
= & \langle \exp (\sum_{n=1}^\infty \frac{p_n}{n} \alpha_n) q^{(a+1)K_0}
\exp \biggl( \sum_{n=1}^\infty \frac{(-1)^{n-1}(1-e^{-nt})}{[n]}
 \frac{\alpha_{-n}}{n} \biggr) \rangle.
\end{split}
\ee
The normalized open string amplitude of the resolved conifold with one outer brane
as in Figure 1(b)
and framing $a \in \bZ$ is given by:
\be \label{eqn:ZbOperator1}
\hat{\tilde{Z}}^{(a)}(\lambda;t; \bp)
=  \langle \exp (\sum_{n=1}^\infty \frac{p_n}{n} \beta_n) q^{aK_0}
\exp \biggl( \sum_{n=1}^\infty \big( \frac{1}{[n]} - \frac{e^{-nt}}{[n]} \big) \frac{\beta_{-n}}{n} \biggr) \rangle.
\ee
These give us two tau-functions of the KP hierarchy.

\subsection{The first case}

The  element in $\Gr^{(0)}$ corresponding to $\hat{Z}^{(a)}$ is spanned by the following series
\be
\begin{split}
\varphi_j = &
\exp \biggl(\frac{\sqrt{-1}(a+1)\lambda}{2}[(z\pd_z+1/2)^2- (j+1/2)^2]\biggr) \\
&
\exp\left( \sum_{n > 0} \frac{(-1)^{n-1}(1-e^{-nt})}{[n]} \frac{z^{-n}}{n} \right) z^j.
\end{split}
\ee

\begin{thm}
The series $\varphi_n$  can be explicitly written as:
\be \label{eqn:Con-Basis}
\varphi_j
= \sum_{n=0}^\infty \frac{(-1)^n q^{(a+1)n(n-2j-1)/2}\prod_{k=1}^n (1- e^{-t} q^{k-1})}{\prod_{k=1}^n (1-q^k)}
q^{n/2} z^{j-n}.
\ee
In particular,
\be
\varphi_0 = \sum_{n=0}^\infty \frac{(-1)^n q^{(a+1)n(n-1)/2}\prod_{k=1}^n (1- e^{-t} q^{k-1})}{\prod_{k=1}^n (1-q^k)}
q^{n/2} z^{-n}.
\ee
This matches with \cite[(45)]{Zhou-Quant}.
\end{thm}

\begin{proof}
Note that
\ben
&&  \exp\left( \sum_{n > 0} \frac{(-1)^{n-1}+(-e^{-t})^n}{n(q^{n/2} -q^{-n/2})} z^{-n} \right) \\
& = & \exp\left( \sum_{n > 0} \frac{(-1)^n(1-e^{-nt})}{n}
\sum_{m=1}^\infty q^{(m-1/2)n} z^{-n} \right) \\
& = & \prod_{m=1}^\infty \frac{1+e^{-t}q^{m-1/2}/z}{1+q^{m-1/2}/z} \\
& = & \sum_{n=0}^\infty (-1)^n \prod_{j=1}^n \frac{1-e^{-t} q^{j-1}}{1-q^j}
\cdot q^{n/2} z^{-n}.
\een
Here in the last equality we use the following well-known identity:
\be \label{eqn:Product}
\prod_{m=0}^\infty \frac{1-aq^{m+1}x}{1-q^m x}
= 1 + \sum_{n\geq 1} \frac{\prod_{j=1}^n (1-aq^j)}{\prod_{j=1}^n (1-q^j)} x^n.
\ee
Therefore,
\ben
\varphi_j & = &
\exp \biggl(\frac{\sqrt{-1}(a+1)\lambda}{2}[(z\pd_z+1/2)^2- (j+1/2)^2]\biggr) \\
&& \sum_{n=0}^\infty (-1)^n \prod_{k=1}^n \frac{1-e^{-t} q^{k-1}}{1-q^k}
\cdot q^{n/2} z^{j-n} \\
& = & \sum_{n=0}^\infty
\exp \biggl(\frac{\sqrt{-1}(a+1)\lambda}{2}[(j-n+1/2)^2- (j+1/2)^2]\biggr) \\
&& \cdot (-1)^n \prod_{k=1}^n \frac{1-e^{-t} q^{k-1}}{1-q^k}
\cdot q^{n/2} z^{j-n} \\
& = & \sum_{n=0}^\infty \frac{(-1)^n q^{(a+1)n(n-2j-1)/2}\prod_{k=1}^n (1- e^{-t} q^{k-1})}{\prod_{k=1}^n (1-q^k)}
q^{n/2} z^{j-n}.
\een
\end{proof}

\begin{thm} \label{thm:Conifold}
Define two operators $\hat{P}_0, \hat{Q}_0$ as follows:
\bea
&& \hat{P}_0 = 1 - q^{- z\pd_z} + q^{1/2}
z^{-1} q^{-(a+1) z \pd_z}
- q^{1/2} e^{-t} z^{-1} q^{- (a+2)  z\pd_z}, \\
&& \hat{Q}_0 = z q^{(a+1) z \pd_z}.
\eea
Then one has:
\bea
&& \hat{Q}_0 \varphi_j = q^{(a+1)j} \varphi_{j+1},  \label{eqn:Qphi} \\
&& \hat{P}_0 \varphi_j = (1-q^{-j}) \varphi_j + q^{1/2- (r+1)j}(1-q^{-j}) \varphi_{j-1}, \label{eqn:Pphi}
\eea
for $j =0, 1, 2, \dots$,
where $\varphi_{-1} = 0$.
In particular,
\be
\hat{P}_0 \varphi_0 = 0.
\ee
This matches with \cite[(49)]{Zhou-Quant}.
Furthermore,
\be \label{eqn:[P,Q]2}
[\hat{P}_0, \hat{Q}_0] = (1-q^{-1}) zq^{a z\pd_z}
+ (1-q^{-1})e^{-t}q^{-1/2-a} q^{-z\pd_z}.
\ee
\end{thm}

\begin{proof}
The equality \eqref{eqn:Qphi} and \eqref{eqn:Pphi} can be checked directly from \eqref{eqn:Con-Basis}.
\ben
\hat{Q}_0 \varphi_j
& = & z e^{(a+1) \lambda z \pd_z} \sum_{n=0}^\infty
\frac{(-1)^n q^{(a+1)n(n-2j-1)/2}\prod_{k=1}^n (1- e^{-t} q^{k-1})}{\prod_{k=1}^n (1-q^k)}
q^{n/2} z^{j-n} \\
& = & \sum_{n=0}^\infty \frac{(-1)^n q^{(a+1)n(n-2j-1)/2}\prod_{k=1}^n (1- e^{-t} q^{k-1})}{\prod_{k=1}^n (1-q^k)}
q^{n/2}
q^{(a+1)(j-n)} z^{j+1-n} \\
& = & q^{(a+1)j} \sum_{n=0}^\infty \frac{(-1)^n q^{(a+1)n(n-2(j+1)-1)/2}\prod_{k=1}^n (1- e^{-t} q^{k-1})}{\prod_{k=1}^n (1-q^k)}
q^{n/2}
 z^{j+1-n} \\
& = & q^{(a+1)j} \varphi_{j+1},
\een
and for \eqref{eqn:Pphi},
first note:
\ben
\hat{P}_0 \varphi_j
& = & \sum_{n=0}^\infty
\frac{(-1)^n q^{(a+1)n(n-2j-1)/2+n/2}\prod_{k=1}^n (1- e^{-t} q^{k-1})}{\prod_{k=1}^n (1-q^k)} z^{j-n} \\
& - & \sum_{n=0}^\infty
\frac{(-1)^n q^{(a+1)n(n-2j-1)/2+n/2}\prod_{k=1}^n (1- e^{-t} q^{k-1})}{\prod_{k=1}^n (1-q^k)} q^{n-j} z^{j-n} \\
& + & q^{1/2} z^{-1} \sum_{n=0}^\infty
\frac{(-1)^n q^{(a+1)n(n-2j-1)/2+n/2}\prod_{k=1}^n (1- e^{-t} q^{k-1})}{\prod_{k=1}^n (1-q^k)} \\
&& \cdot  q^{(r+1)(n-j)} z^{j-n} \\
& - & q^{1/2}e^{-t} z^{-1} \sum_{n=0}^\infty
\frac{(-1)^n q^{(a+1)n(n-2j-1)/2+n/2}\prod_{k=1}^n (1- e^{-t} q^{k-1})}{\prod_{k=1}^n (1-q^k)} \\
&& \cdot  q^{(r+2)(n-j)} z^{j-n}.
\een
The first term on the right-hand side is just $\varphi_j$,
the second term on the right-hand side can be rewritten as follows:
\ben
&&- q^{-j} \varphi_j  \\
& + & q^{-j} \sum_{n=1}^\infty
\frac{(-1)^n q^{(a+1)n(n-2j-1)/2+n/2}\prod_{k=1}^n (1- e^{-t} q^{k-1})}{\prod_{k=1}^{n-1} (1-q^k)} z^{j-n} \\
& = & - q^{-j} \varphi_j  \\
& + & q^{-j} \sum_{n=0}^\infty
\frac{(-1)^{n+1} q^{(a+1)(n+1)(n-2j)/2+(n+1)/2}\prod_{k=1}^{n+1} (1- e^{-t} q^{k-1})}{\prod_{k=1}^{n} (1-q^k)} z^{j-1-n} \\
& = & - q^{-j} \varphi_j  \\
& - & q^{-j} q^{1/2-(a+1)j} \sum_{n=0}^\infty
\frac{(-1)^n q^{(a+1)n(n-2j+1)/2+n/2}\prod_{k=1}^{n} (1- e^{-t} q^{k-1})}{\prod_{k=1}^{n} (1-q^k)} z^{j-1-n} \\
& + & q^{1/2-(r+2)j} e^{-t} \sum_{n=0}^\infty
\frac{(-1)^{n} q^{(a+1)n(n-2j+1)/2+n/2} q^n \prod_{k=1}^{n} (1- e^{-t} q^{k-1})}{\prod_{k=1}^{n} (1-q^k)} z^{j-1-n} \\
& = & - q^{-j} \varphi_j  - q^{-j} q^{1/2-(r+1)j} \varphi_{j-1} \\
& + & q^{1/2-(a+2)j} e^{-t} \sum_{n=0}^\infty
\frac{(-1)^{n} q^{(r+1)n(n-2j+1)/2+n/2} q^n \prod_{k=1}^{n} (1- e^{-t} q^{k-1})}{\prod_{k=1}^{n} (1-q^k)} z^{j-1-n};
\een
the third and the fourth terms on the right-hand side can be rewritten as follows:
\ben
&& q^{1/2-(a+1)j} \sum_{n=0}^\infty
\frac{(-1)^n q^{(r+1)n(n-2j+1)/2+n/2}\prod_{k=1}^n (1- e^{-t} q^{k-1})}{\prod_{k=1}^n (1-q^k)} z^{j-1-n} \\
& - & q^{1/2}e^{-t} \sum_{n=0}^\infty
\frac{(-1)^n q^{(a+1)n(n-2j-1)/2+n/2}\prod_{k=1}^n (1- e^{-t} q^{k-1})}{\prod_{k=1}^n (1-q^k)} q^{(r+2)(n-j)} z^{j-1-n} \\
& = &  q^{1/2-(a+1)j}   \varphi_{j-1} \\
& - & q^{1/2-(a+2)j}e^{-t} \sum_{n=0}^\infty
\frac{(-1)^n q^{(a+1)n(n-2j+1)/2+n/2}\prod_{k=1}^n (1- e^{-t} q^{k-1})}{\prod_{k=1}^n (1-q^k)} q^{n} z^{j-1-n}.
\een
Then \eqref{eqn:Pphi} is proved by putting all these terms together.
The commutation relation can be obtained
by a computation similar to the proof of \eqref{eqn:[P,Q]2}.
\end{proof}

\subsection{The second case}

Similarly, the  element in $\Gr^{(0)}$  corresponding to $\hat{\tilde{Z}}^{(a)}$
is spanned by the following series
\be
\begin{split}
\tilde{\varphi}_j = &
\exp \biggl(\frac{\sqrt{-1} a \lambda}{2}[(z\pd_z+1/2)^2- (j+1/2)^2]\biggr) \\
&
\exp\left( \sum_{n > 0} \frac{(1-e^{-nt})}{[n]} \frac{z^{-n}}{n} \right) z^j.
\end{split}
\ee

\begin{thm}
The series $\tilde{\varphi}_n$  can be explicitly written as:
\be \label{eqn:Con-Basis}
\tilde{\varphi}_j
= \sum_{n=0}^\infty \frac{q^{an(n-2j-1)/2}\prod_{k=1}^n (e^{-t} - q^{k-1})}{\prod_{k=1}^n (1-q^k)}
q^{n/2} z^{j-n}.
\ee
In particular,
\be
\tilde{\varphi}_0 = \sum_{n=0}^\infty \frac{q^{an(n-1)/2}\prod_{k=1}^n (e^{-t} - q^{k-1})}{\prod_{k=1}^n (1-q^k)}
q^{n/2} z^{-n}.
\ee
\end{thm}

\begin{proof}
Note that
\ben
&&  \exp\left( \sum_{n > 0} \frac{1-e^{-nt}}{n(q^{n/2} -q^{-n/2})} z^{-n} \right) \\
& = & \exp\left( \sum_{n > 0} \frac{1-e^{-nt}}{n}
\sum_{m=1}^\infty q^{-(m-1/2)n} z^{-n} \right) \\
& = & \prod_{m=1}^\infty \frac{1-e^{-t}q^{-(m-1/2)}/z}{1-q^{-(m-1/2)}/z} \\
& = & \sum_{n=0}^\infty  \prod_{j=1}^n \frac{1-e^{-t} q^{-(j-1)}}{1-q^{-j}}
\cdot q^{-n/2} z^{-n} \\
& = & \sum_{n=0}^\infty  \prod_{j=1}^n \frac{e^{-t} - q^{j-1}}{1-q^j}
\cdot q^{n/2} z^{-n}.
\een
Here in the last equality we use \eqref{eqn:Product}.
Therefore,
\ben
\tilde{\varphi}_j & = &
\exp \biggl(\frac{\sqrt{-1}a\lambda}{2}[(z\pd_z+1/2)^2- (j+1/2)^2]\biggr) \\
&& \sum_{n=0}^\infty \prod_{k=1}^n \frac{e^{-t} - q^{k-1}}{1-q^k}
\cdot q^{n/2} z^{j-n} \\
& = & \sum_{n=0}^\infty \frac{q^{an(n-2j-1)/2}\prod_{k=1}^n (e^{-t} -q^{k-1})}{\prod_{k=1}^n (1-q^k)}
q^{n/2} z^{j-n}.
\een
\end{proof}

\begin{thm}
Define two operators $\hat{\tilde{P}}_0, \hat{\tilde{Q}}_0$ as follows:
\bea
&& \hat{\tilde{P}}_0 = 1 - q^{- z\pd_z} + q^{1/2}
z^{-1} q^{-(a +1) z \pd_z} - q^{1/2} e^{-t} z^{-1} q^{- a  z\pd_z}, \\
&& \hat{\tilde{Q}}_0 = z q^{a z \pd_z}.
\eea
Then one has:
\bea
&& \hat{\tilde{Q}}_0 \tilde{\varphi}_j = q^{aj} \tilde{\varphi}_{j+1},  \\
&& \hat{\tilde{P}}_0 \tilde{\varphi}_j = (1-q^{-j}) \tilde{\varphi}_j
- e^{-t} q^{1/2- a j}(1-q^{-j}) \tilde{\varphi}_{j-1},
\eea
for $j =0, 1, 2, \dots$,
where $\varphi_{-1} = 0$.
In particular,
\be
\hat{\tilde{P}}_0 \tilde{\varphi}_0 = 0.
\ee
Furthermore,
\be
[\hat{\tilde{P}}_0, \hat{\tilde{Q}}_0] = (1-q^{-1}) zq^{(a -1)z\pd_z}
- (1-q^{-1})q^{1/2-a} q^{-z\pd_z}.
\ee
\end{thm}

The proof is similar to that of Theorem \ref{thm:Conifold}.

\subsection{Quantum torus}

Note both the Lie algebra generated by $\hat{P}_0$ and $\hat{Q}_0$
(in this Section and in last Section)
and the Lie algebra generated by $\hat{\tilde{P}}_0$ and $\hat{\tilde{Q}}_0$
are Lie subalgebras of the Lie algebra associated with the noncommutative associative algebra
spanned by $\{z^m q^{n z\pd_z}\}_{m,n \in \bZ}$.
We have
\be
z^{m_1} q^{n_1 z\pd_z} \cdot z^{m_2} q^{n_2z\pd_z}
= q^{n_1m_2} \cdot z^{m_1+m_2} q^{(n_1+n_2) z\pd_z}.
\ee
This algebra is the algebra of two-dimensional quantum complex torus.
This fact indicates that the relevant quantum spectral curves
live in the two-dimensional quantum complex torus,
compatible with the fact that the corresponding local mirror curves 
live in the two-torus.

\vspace{.1in}

{\em Acknoledgements}.
This research is partially supported by NSFC grant 11171174.

\end{document}